\documentclass[journal]{IEEEtranTCOM}

\usepackage{mathrsfs}
\usepackage[dvips]{graphicx}
\usepackage[usenames,dvipsnames]{color}
\usepackage{fancyhdr}
\usepackage{amssymb,amsmath}
\usepackage{lastpage}
\usepackage{epsfig}
\usepackage{algorithm2e}
\usepackage[compress]{cite}
\usepackage[caption=false,font=footnotesize]{subfig}
\usepackage{url,balance}

\newtheorem{theorem}{Theorem}
\newtheorem{corollary}[theorem]{Corollary}

\newfont{\mbb}{msbm10 scaled 1100}

\definecolor{red}{rgb}{0.9,0.0,0.1}

\definecolor{blue}{rgb}{.15,.4,.8}

\definecolor{green}{rgb}{.25,.8,.25}

\definecolor{orange}{RGB}{255,128,0}

\definecolor{dgreen}{RGB}{0,151,0}

\def\Pr{{\rm Pr}} 

\def\b0{{\bf 0}}

\title{Localized Dimension Growth: A
Convolutional Random Network Coding Approach to Managing Memory and
Decoding Delay}
\author{Wangmei~Guo,
        Xiaomeng~Shi,~\IEEEmembership{Student~Member,~IEEE,}
        Ning~Cai,~\IEEEmembership{Member,~IEEE,}
        and Muriel~M\'{e}dard,~\IEEEmembership{Fellow,~IEEE}
\thanks{This work has been partially presented at ISIT 2011.}
\thanks{Wangmei Guo and Ning Cai are with The State Key Laboratory of ISN, Xidian University, Xi'an, China. email: \{wangmeiguo, caining\}@mail.xidian.edu.cn. This material is based upon work supported by the National Natural Science Foundation of China under Grant No. 60832001.}%
\thanks{Xiaomeng Shi and Muriel M\'{e}dard are with the Research Laboratory of Electronics, Department of Electrical Engineering and Computer Science, Massachusetts Institute of Technology, Cambridge, MA, USA. e-mail: \{xshi, medard\}@mit.edu. This material is based upon work supported by the Air Force Office of Scientific Research (AFOSR) under award number 016974-002 S, the Claude E. Shannon Research Assitantship from RLE, and by the NSERC Postgraduate Scholarship (PGS) issued by the Natural Sciences and Engineering Research Council of Canada.}}

\normalsize
\begin{document}
\maketitle
\begin{abstract}
We consider an \textit{Adaptive Random Convolutional Network Coding} (ARCNC) algorithm to address the issue of field size in random network coding for multicast, and study its memory and decoding delay performances through both analysis and numerical simulations. ARCNC operates as a convolutional code, with the coefficients of local encoding kernels chosen randomly over a small finite field. The cardinality of local encoding kernels increases with time until the
global encoding kernel matrices at related sink nodes have full rank. 
ARCNC adapts to unknown network topologies without prior knowledge, by locally incrementing the dimensionality of the convolutional code. 
Because convolutional codes of different constraint lengths can coexist in different portions of the network, reductions in decoding delay and memory overheads can be achieved. We show that this method performs no worse than random linear network codes in terms of decodability, and can provide significant gains in terms of average decoding delay or memory in combination, shuttle and random geometric networks.

\end{abstract}

\begin{IEEEkeywords}
convolutional network codes, random linear network codes, adaptive
random convolutional network code, combination networks, random
graphs.
\end{IEEEkeywords}

%
\section{Introduction}\label{sec:introduction}
\IEEEPARstart{S}{ince} its introduction \cite{ACLY2000}, network coding has been shown to offer advantages in throughput, power consumption, and security in   wireline and wireless networks. Field size and adaptation to unknown topologies are two of the key issues in network coding. Li et al. showed constructively that the max-flow bound is achievable by linear algebraic network coding (ANC) if the
field is sufficiently large for a given deterministic multicast network \cite{LYC2003}, while Ho et~al. \cite{HMK2006} proposed a distributed random linear network code (RLNC) construction that achieves the multicast capacity with probability $(1-d/q)^\eta$, where $\eta$ is the number of links with random coefficients, $d$ is the number of sinks, and $q$ is the field
size. Because of its construction simplicity and the ability to adapt to unknown topologies, RLNC is often preferred over deterministic network codes. While the construction in \cite{HMK2006} allows cycles, which leads to the creation of
convolutional codes, it does not make use of the convolutional nature of the resulting codes to lighten bounds on field size, which may need to be large to guarantee decoding success at all sinks. Both block network codes (BNC) \cite{medard2003coding,JSCEEJT2005} and convolutional network codes
(CNC) \cite{OCNC2006,EF2004} can mitigate field size requirements. M\'edard et al. introduced the concept of BNC \cite{medard2003coding}; Xiao et al. proposed a deterministic binary BNC to solve the combination network problem \cite{xiao2008binary}. BNC can operate on smaller finite fields, but the block length may need to be pre-determined according to network size. In discussing cyclic networks, both Li et al. and Ho et al. pointed out the equivalence between ANC in cyclic networks with delays, and CNC \cite{LYC2003,HMK2006}. Because of coding introduced across the temporal domain, CNC in general does not have a field size constraint.

Combining the adaptive and distributive advantages of RLNC and the field-size independence of CNC, we proposed adaptive random convolutional network code (ARCNC) in \cite{guo2011localized} as a localized coding scheme for single-source multicast. ARCNC randomly chooses local encoding kernels from a small field, and the code constraint length increases locally at each node. In general, sinks closer to the source adopts a smaller code length than that of sinks far away. ARCNC adapts to unknown network topologies without prior knowledge, and allows convolutional codes with different code lengths to coexist in different portions of the network, leading to reduction in decoding delay and memory overheads associated with using a pre-determined field size or code length.

Concurrently to \cite{guo2011localized}, Ho et al. proposed a variable length CNC \cite{ho2010universal} and provided a mathematical proof to show that the overall error probability of the code can be bounded when each intermediate node chooses its code length from a range estimated from its depth. The encoding process involves a graph transformation of the network into a ``low-degree'' form, with each node having a degree of at most 3. Our work differs from \cite{ho2010universal} in that our approach uses feedbacks algorithmically.

In this paper, we first describe the ARCNC algorithm in acyclic and
cyclic networks and show that ARCNC converges in a finite amount of
time with probability 1. We then provide several examples to
illustrate the decoding delay and memory gains ARCNC offers in
deterministic and random networks. Our first example is $n\choose m$
combination networks. Ngai and Yeung have previously pointed out
that throughput gains of network coding can be unbounded over
combination networks \cite{ngai2004network}. Our analysis shows that
the average decoding delay is bounded by a constant when $m$ is
fixed and $n$ increases. In other words, the decoding delay gain
becomes infinite as the number of intermediate nodes increases in a
combination network. On the other hand, our numerical simulation
shows that the decoding delay increases sublinearly when $m=n/2$ and
$n$ increases in value. We then consider a family of networks
defined as \emph{sparsified combination networks} to illustrate the
effect of interdependencies among sinks and depth of the network on
memory use. For cyclic networks, we consider the shuttle network as
an example. We also extend the application of ARCNC from structured
cyclic and acyclic networks to random geometric graphs, where we
provide empirical illustration of the benefits of ARCNC.

The remainder of this paper is organized as follows: the ARCNC algorithm is proposed in Section~\ref{sec:algorithm}; performance analysis is given in Section~\ref{sec:analysis}. The coding delay and memory advantages of ARCNC are discussed for combination and shuttle networks in Section~\ref{sec:examples}. Numerical results are provided in Section~\ref{sec:simulations} for combination and random networks. Section~\ref{sec:conclusion} concludes the paper.

\section{Adaptive Randomized Convolutional Network Codes}\label{sec:algorithm}
\subsection{Basic Model and Definitions}\label{sec:basicDefs}

We model a communication network as a finite directed multigraph,
denoted by $\mathcal{G}=(\mathcal{V},\mathcal{E})$, where
$\mathcal{V}$ is the set of nodes and $\mathcal{E}$ is the set of
edges. An edge represents a noiseless communication channel with
unit capacity. We consider the single-source multicast case, i.e.,
the source sends the same messages to all the sinks in the network.
The source node is denoted by $s$, and the set of $d$ sink nodes is
denoted by $R=\{r_1,\ldots,r_d\} \subset \mathcal{V}$. For every
node $v \in \mathcal{V}$, the sets of incoming and outgoing channels
to $v$ are $In(v)$ and $Out(v)$; let $In(s)$ be the empty set
$\emptyset$. An ordered pair $(e',e)$ of edges is called an
\emph{adjacent pair} when there exists a node $v$ with $e'\in In(v)$
and $e\in Out(v)$. Since edges are directed, we use the terms edge
and arc interchangeably in this paper.

The symbol alphabet is represented by a base field, $\mathbb{F}_q$.
Assume $s$ generates a source \emph{message} per unit time,
consisting of a fixed number of $m$ source \emph{symbols}
represented by a size $m$ \emph{row} vector $x_t
=(x_{1,t},x_{2,t},\cdots,x_{m,t})$, $x_{i,t}\in\mathbb{F}_q$. Time
$t$ is indexed from 0, with the $(t+1)$-th message is
generated at time $t$. The source messages can be collectively
represented by a power series $x(z)=\sum_{t\geq 0}{x_t z^{t}}$,
where $x_t$ is the message generated at time $t$ and $z$ denotes a
unit-time delay. $x(z)$ is therefore a row vector of polynomials
from the ring $\mathbb{F}_q[z]$.

Denote the data propagated over a channel $e$ by $y_e(z)=\sum_{t\geq
0}{y_{e,t} z^{t}}$, where $y_{e,t}\in \mathbb{F}_q$ is the data
symbol sent on edge $e$ at time $t$. For edges connected to the
source, let $y_e(t)$ be a linear function of the source messages,
i.e., for all $e\in Out(s)$, $y_e(z) = x(z)f_e(z)$, where
$f_e(z)=\sum_{t\geq0}f_{e,t}z^{t}$ is a size $m$ column vector of
polynomials from $\mathbb{F}_q[z]$. For edges not connected directly
to the source, let $y_e(z)$ be a linear function of data transmitted
on incoming adjacent edges $e'$, i.e., for all $v \neq s$, $e\in
Out(v)$,
\begin{align}
  y_e(z) = \sum_{e' \in In(v)}{k_{e',e}(z)y_{e'}(z)}\,. \label{eq:yet_k}
\end{align}
Both $k_{e',e}(z)$ and $y\,_e(z)$ are in $\mathbb{F}_q[z]$. Define
$k_{e',e}(z)= \sum_{t\geq0}k_{e', e, t}z^{t}$ as the \emph{local
encoding kernel} over the adjacent pair $(e',e)$, where
$k_{e',e,t}\in \mathbb{F}_q$. Thus, for all $e\in \mathcal{E}$,
$y_e(z)$ is a linear function of the source messages, \begin{align}
   y_e(z)= x(z) f_e(z)\,, \label{eq:yet_f}
\end{align}
\noindent where $f_e(z)=\sum_{t\geq0}f_{e,t}z^{t}$ is the size $m$ \emph{column} vector defined as the \emph{global encoding kernel} over channel $e$, and for all $v \neq s$, $e\in Out(v)$,
\begin{align}
  f_e(z) & = \sum_{e' \in In(v)}{k_{e',e}(z)f_{e'}(z)}\label{eq:fe}\,, \\
\text{i.e.\,,} \quad \quad f_{e,t} & = \sum_{e'\in In(v)}\left(\sum_{i=0}^t k_{e',e,i}f_{e',t-
i}\right)\,. \label{eq:fet}
\end{align}
Note that $f_{e,t}\in \mathbb{F}_q^m$, and $f_{e'}(z), f_e(z)\in
\mathbb{F}_q^m[z]$. Expanding Eq.~\eqref{eq:yet_k} term by term
gives an explicit expression for each data symbol $y_{e,t}$
transmitted on edge $e$ at time $t$, in terms of source symbols and
global encoding kernel coefficients:
\begin{align}
y_{e,t} & = \sum_{e'\in In(v)}\left(\sum_{i=0}^t k_{e',e,i}y_{e',t-
i}\right) = \sum\limits_{i=0}^t{x_{t-i}f_{e,i}}\,. \label{eq:yet}
\end{align}
Each intermediate node $v\neq s$ is therefore required to store in
its memory received data symbols $y_{e',t-i}$ for values of $i$ at
which $k_{e',e,i}$ is non-zero. The design of a CNC is the process of determining local encoding kernel coefficients $k_{e',e,t}$ for all adjacent pairs $(e',e)$, and $f_{e,t}$ for $e\in Out(s)$, such that the original source messages
can be decoded correctly at the given set $R$ of sink nodes. With a
random linear code, these coding kernel coefficients are chosen
uniformly randomly from the finite field $\mathbb{F}_q$. This paper
studies an adaptive scheme where kernel coefficients are generated
one at a time until decodability is achieved at all sinks.

Collectively, we call the $|In(v)|\times|Out(v)|$ matrix
$K_v(z)=(k_{e',e}(z))_{e'\in In(v), e\in
Out(v)}=K_{v,0}+K_{v,1}z+K_{v,2}z^2+\ldots$ the \emph{local encoding
kernel matrix} at node $v$, and the $m \times |In(v)|$ matrix
$F_v(z)=(f_{e}(z))_{e\in In(v)}$ the \emph{global encoding kernel
matrix} at node $v$. Observe from Eq.~\eqref{eq:yet_f} that, at sink
$r$, $F_{r}(z)$ is required to deconvolve the received data messages
$y_{e^\prime_i}(z)$, $e^\prime_i \in In(r)$. Therefore, each
intermediate node $v$ computes $f_e(z)$ for outgoing edges from
$F_v(z)$ according to Eq.\eqref{eq:fe}, and sends $f_e(z)$ along
edge $e$, together with data $y_e(z)$. This can be achieved by
arranging the coefficients of $f_e(z)$ in a vector form and
attaching them to the data. In this paper, we ignore the effect of
this overhead transmission of coding coefficients on throughput or
delay: we show in Section~\ref{sec:stoppingTime} that the number of
terms in $f_e(z)$ is finite, thus the overhead can be amortized over
a long period of data transmissions.

Moreover, $F_v(z)$ can be written as
$F_v(z)=F_{v,0}+F_{v,1}z+\cdots+F_{v,t}z^{t}$, where $F_{v,t}\in
\mathbb{F}_q^{m\times In(v)}$ is the global encoding kernel matrix
at time $t$. $F_v(z)$ can thus be viewed as a polynomial, with
$F_{v,t}$ as matrix coefficients. Let $L_v$ be the degree of
$F_v(z)$. $L_v+1$ is a direct measure of the amount of memory
required to store $F_v(z)$. We shall define in
Section~\ref{sec:memory} the metric used to measure memory overhead
of ARCNC.

\subsection{Algorithm for Acyclic Networks}\label{sec:algAcyclic}

\subsubsection{Code Generation and Data Encoding} initially, all local and global encoding kernels are set to 0. At time $t$, the $(t+1)$-th coefficient $k_{e',e,t}$ of the local encoding kernel $k_{e',e}(z)$ is chosen uniformly randomly from $\mathbb{F}_q$ for each adjacent pair $(e',e)$, independently from other kernel coefficients. Each node $v$ stores the local encoding kernels and forms the outgoing data symbol as a random linear combination of incoming data symbols in its memory according to  Eq.~\eqref{eq:yet}. Node $v$ also stores the global encoding kernel matrix $F_v(z)$ and computes  the global encoding kernel $f_{e}(z)$, in the form of a vector of coding coefficients, according to Eq.~\eqref{eq:fe}.
During this code construction process, $f_e(z)$ is attached to the data transmitted on $e$. Once code generation terminates and the CNC $F_r(z)$ is known at each sink $r$, $f_e(z)$ no longer needs to be forwarded, and only data symbols are sent on each outgoing edge. Recall that we ignore the reduction in rate due to the transmission
of coding coefficients, since this overhead can be amortized over long periods of data transmissions.

In acyclic networks, a complete topological order exists among the nodes, starting from the source. Edges can be ranked such that coding can be performed sequentially, where a downstream node encodes after all its upstream nodes have generated their coding coefficients. Observe that we have not assumed non-zero transmission delays.

\subsubsection{Testing for Decodability and Data Decoding} \label{subsubsec:decoding}
at every time instant $t$, each sink $r$ decides whether its global encoding kernel matrix $F_r(z)$ is full rank. If so, it sends an ACK signal to its parent node. An intermediate node $v$ which has received ACKs from all its children at time $t_0$ will send an ACK to its parent, and set all subsequent local encoding kernel coefficients $k_{e',e,t}$ to $0$ for all $t>t_0$, $e' \in In(v)$, and $e\in Out(v)$. In other words, the constraint lengths of the
local convolutional codes increase until they are sufficient for downstream sinks to decode successfully. Such automatic adaptation eliminates the need for estimating the field size or the constraint length a priori. It also allows nodes within the network to operate with different constraint lengths as needed.

If $F_r(z)$ is not full rank, $r$ stores received messages and waits for more data to arrive. At time $t$, the algorithm is considered successful if all sinks can decode. This is equivalent to saying that the determinant of $F_r(z)$ is a non-zero polynomial. Recall from Section~\ref{sec:basicDefs}, $F_r(z)$ can be written as $F_r(z)=F_{r,0}+F_{r,1}z+\cdots+F_{r,t}z^{t}$, where $F_{r,t}$ is
the global encoding kernel matrix at time $t$. Computing the determinant of $F_r(z)$ at every time instant $t$ is complex, so we test instead the following two conditions, introduced in \cite{CG2009} and \cite{MS1968} to determine decodability at a sink $r$. The first condition is necessary and easy to compute, while the second is both necessary and sufficient, but slightly more complex.

\begin{enumerate}
\item $rank(\widehat{F}_{r,t})=m$, where $\widehat{F}_{r,t}=(F_{r,0} , F_{r,1}, \ldots, F_{r,t})$.
\item $rank(M_{r,t})-rank(M_{r,t-1})=m$, where
\begin{align}
     M_{r,i} = \left( {\begin{array}{*{20}c}
     F_{r,0} & F_{r,1} & \cdots  & F_{r,i} \\
       0     & \ddots  & \ddots  & \vdots  \\
       0     & \cdots  & F_{r,0} & F_{r,1} \\
       0     & \cdots  &    0    & F_{r,0} \\
\end{array}} \right).
\end{align}
\end{enumerate}

Once $F_r(z)$ is full rank, $r$ can perform decoding operations. Let
$T_r$ be the \emph{first decoding time}, or the earliest time at
which the decodability conditions are satisfied. Denote by
${x}_0^{T_r}$ and ${y}_0^{T_r}$ the row vectors $(x_0, \cdots,
x_{T_r})$ and $(y_0, \cdots, y_{T_r})$. Each source message $x_t$ is
a size $m$ row vector of source symbols $x_{i,t}\in \mathbb{F}_q$
generated at $s$ at time $t$; each data message $y_t$ is a size
$In(r)$ row vector of data symbols $y_{e,t}\in \mathbb{F}_q$
received on the incoming edges of $r$ at time $t$, $e\in In(r)$.
Hence, ${y}_0^{T_r} = {x}_0^{T_r} M_{T_r}$. To decode, we want to
find a size $In(r)({T_r}+1)\times m$ matrix $D$
such that
$M_{T_r}D=
{I_{m} \choose \textbf{0}} $. We can then recover source message
$x_0$ by evaluating $y_0^{T_r} D= x_0^{T_r}M_{T_r} D= x_0$. Once $D$
is determined, we can decode sequentially the source message $x_t$
at time $t+{T_r}$, $t>0$. Note that if $|In(r)|>m$, we can simplify
the decoding process by using only $m$ independent received symbols
from the $|In(r)|$ incoming edges.

Observe that, an intermediate node $v$ only stops lengthening its
local encoding kernels $k_{e',e}(z)$ when \emph{all} of its
downstream sinks achieve decodability. Thus, for a sink $r$ with
first decoding time ${T_r}$, the length of $F_r(z)$ can increase
even after $T_r$. Recall from Section~\ref{sec:basicDefs} that $L_r$
is the degree of $F_r(z)$. We will show in
Section~\ref{sec:stoppingTime} that ARCNC converges in a finite
amount of time for a multicast connection. In other words, when the
decodability conditions are satisfied at \emph{all} sinks, the
values of $L_r$ and $T_r$ at an individual sink $r$ satisfy the
condition $L_r\geq {T_r}$, where $L_r$ is finite. Decoding of
symbols after time ${T_r}$ can be conducted sequentially. Details of
the decoding operations can be found in \cite{guoLetter2012}.

\subsubsection{Feedback}
As we have described in the decoding subsection, acknowledgments are
propagated from sinks through intermediate nodes to the source to indicate if code length should continue to be increased at coding nodes. ACKs are assumed to be instantaneous and require no dedicated network links, thus incurring no additional delay or throughput costs. Such assumptions may be reasonable in many systems since feedback is only required during the code construction process. Once code length adaptation finishes, ACKs are no longer needed. We show in Section~\ref{sec:stoppingTime} that ARCNC terminates in a finite amount of time. Therefore, the cost of feedback can be amortized over periods of data transmissions.

\subsection{Algorithm Statement for Cyclic Networks}\label{sec:algCyclic}

In an acyclic network, the local and global encoding kernel descriptions of a linear network code are equivalent, in the sense that for a given set of local encoding kernels, a set of global encoding kernels can be calculated recursively in any upstream-to-downstream order. In other words, a code generated from local encoding kernels has a unique solution when decoding is performed on the corresponding global encoding kernels. By comparison, in a cyclic network, partial orderings of edges or nodes are not always consistent. Given a set of local encoding kernels, there may exist a unique, none, or multiple sets of global encoding kernels (\S 3.1, \cite{NWT2005}). If the code is non-unique, the decoding process at a sink may fail. A sufficient condition for a CNC to be successful is that the constant coefficient matrix consisting of all local encoding kernels be nilpotent \cite{KM2003}; this condition is satisfied if we code over an acyclic topology at $t=0$ \cite{CG2009}. In the extreme case, all local encoding kernels can be set to 0 at $t=0$. This setup translates to a unit transmission delay on each link, which as previous work on RLNC has shown, guarantees the uniqueness of a code construction \cite{HMK2006}. To minimize decoding delay, it is intuitive to make as few local encoding kernels zero as possible. In other words, a reasonable heuristic is to assign 0 to a minimum number of $k_{e',e,0}$, $(e',e)\in\mathcal{E}$, and to assign values chosen uniformly randomly from $\mathbb{F}_q$ to the rest. The goal is to guarantee that each cycle contains at least a single delay.

Although seemingly similar, this process is actually not the same as the problem of finding the minimal feedback edge set. A feedback edge set is a set containing at least one edge of every cycle in the graph. When a feedback edge set is removed, the graph becomes an acyclic directed graph. In our setup, however, since $k_{e',e,0}$ is specific to an adjacent edge pair, $k_{e',e,0}$ does not need to be 0 for all $e'$ where $(e',e) \in \mathcal{V}$.

\begin{figure}[t!]
  \centering
  \includegraphics[width=7.5cm]{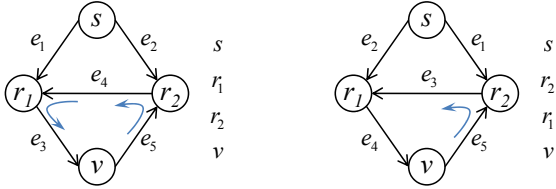}\\
  \caption{A sample cyclic network with edges numerically indexed. The set of indices is not unique, and depends on the order at which nodes are visited, starting from $s$. On the left, $r_1$ is visited before $r_2$; on the right, $r_2$ is visited before $r_1$. In each case, $(e',e)$ is highlighted with a curved arrow if $e'\succeq e$.}\label{fig:cyclicExample}
\end{figure}
For example, a very simple but not necessarily delay-optimal heuristic is to index all edges, and to assign 0 to $k_{e',e,0}$ if $e'\succeq e$, i.e., when $e'$ has an index larger than $e$; $k_{e',e,0}$ is chosen randomly from $\mathbb{F}_q$ if $e' \prec e$. Fig.~\ref{fig:cyclicExample} illustrates this indexing scheme. A node is considered to be visited if one of its incoming edges has been indexed; a node is put into a queue once it is visited. For
each node removed from the queue, numerical indices are assigned to
all of its outgoing edges. Nodes are traversed starting from the
source $s$. The outgoing edges of $s$ are therefore numbered from 1
to $|Out(s)|$. Note that the index set thus obtained is not
necessarily unique. In this particular example, we can have two
sets of edge indices, as shown in Fig.~\ref{fig:cyclicExample}. Here $r_1$ is visited before $r_2$ on the left, and vice versa on the right. In each case, an adjacent
pair $(e',e)$ is highlighted with a curved arrow if $e'\succeq e$.
At $t=0$, we set $k_{e',e,0}$ to 0 for such highlighted adjacent
pairs, and choose $k_{e',e,0}$ uniformly randomly from
$\mathbb{F}_q$ for other adjacent pairs.

Observe that, in an acyclic network, this indexing scheme provides a total ordering for the nodes as well as for the edges: a
node is visited only after all of its parents and ancestors are
visited; an edge is indexed only after all edges on any of its paths
from the source are indexed. In a cyclic network, however, an order
of nodes and edges is only partial, with inconsistencies around each
cycle. Such contradictions in the partial ordering of edges make the
generation of unique network codes along each cycle impossible. By
assigning 0 to local encoding kernels $k_{e',e,0}$ for which
$e'\succeq e$, such inconsistencies can be avoided at time 0, since
the order of $e'$ and $e$ becomes irrelevant in determining the
network code. After the initial step, $k_{e',e,t}$ is not
necessarily 0 for $e'\succeq e$, $t>0$, nonetheless the convolution
operations at intermediate nodes ensure that the 0 coefficient
inserted at $t=0$ makes the global encoding kernels unique at the
sinks. This idea can be derived from the expression for
$f_{e,t}$ given in Eq.~\eqref{eq:fet}. In each cycle, there is at
least one $k_{e',e,0}$ that is equal to zero. The corresponding
$f_{e',t}$ therefore does not contribute to the construction of
other $f_{e,t}$'s in the cycle. In other words, the partial ordering
of arcs in the cycle can be considered consistent at $t=1$ and later
times.

Although this heuristic for cyclic networks is not optimal, it is universal. One disadvantage of this approach is that full knowledge of the topology is required at $t=0$, making the algorithm centralized instead of entirely distributed. Nonetheless, if inserting an additional transmission delay on each link is not an issue, we can always bypass this code assignment stage by zeroing all local encoding kernels at $t=0$.  

After initialization, the algorithm proceeds in exactly the same way as in the acyclic case.

\section{Analysis}\label{sec:analysis}
\subsection{Success probability}

Discussions in \cite{HMK2006, LYC2003, KM2003} state that in a network with delays, ANC gives rise to random processes which can be written algebraically in terms of a delay variable $z$. Thus, a convolutional code can naturally evolve from message propagation and linear encoding. ANC in the delay-free case is therefore equivalent to CNC with constraint length 1. Similarly, using a CNC with constraint length $l>1$ on a delay-free network is equivalent to performing ANC on the same network, but with $l-1$ self-loops attached to each encoding node. Each self-loop carries $z, z^2, \ldots, z^{l-1}$ units of delay respectively.

\begin{figure}[t!]
  \centering
  \includegraphics[width=3.5in]{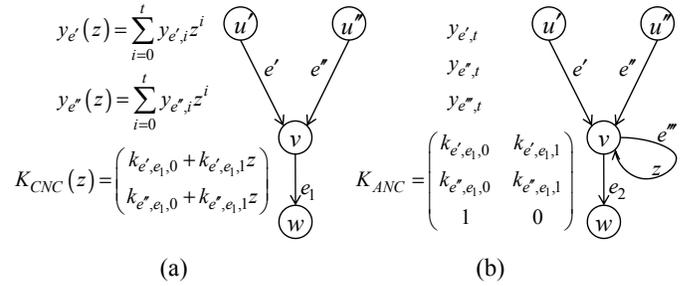}
 \caption{A convolution code resulting from self-loops on a network with transmission delays. (a) CNC in a delay-free network. Data transmitted on incoming edges are $y_{e'}(z)$ and $y_{e}(z)$ respectively. The local encoding kernels are given by $K_{\text{CNC}}(z)$. (b) Equivalent ANC in a network with delays. The given self-loop carries a single delay $z$. Incoming data symbolsare $y_{e',t}$ and $y_{e'',t}$ at time $t$. The ANC coding coefficients are given by the matrix $K_{\text{ANC}}$.}\label{fig:ANC_CNC}
\end{figure}

For example, in Fig.~\ref{fig:ANC_CNC}, we show a node with two incoming edges. Let the data symbol transmitted on edge $\dot{e}$ at time $t$ be $y_{\dot{e},t}$. A CNC with length $l=2$ is used in (a), assuming that transmissions are delay-free. The local encoding kernel matrix $K_{\text{CNC}}(z)$ contains two polynomials, $k_{e',e_1}(z) = k_{e',e_1,0}+k_{e',e_1,1}z$ and $k_{e'',e_1}(z) = k_{e'',e_1,0}+k_{e'',e_1,1}z$. According to Eq.~\eqref{eq:yet_k} and \eqref{eq:yet}, the data symbol transmitted on $e_1$ at time $t$ is
\begin{align}
y_{e_1,t} & = \sum_{\dot{e}\in\{e',e''\}} y_{\dot{e},t}k_{\dot{e},e_1,0}+y_{\dot{e},t-1}k_{\dot{e},e_1,1}\,. \label{ye1tCNC}
\end{align}
In (b), the equivalent ANC is shown. A single loop with a transmission delay of $z$ has been added, and the local encoding kernel matrix $K_{\text{ANC}}=(k_{\dot{e},e})_{\dot{e}\in In(v),e\in Out(v)}$ is constructed from coding coefficients from (a). The first column of $K_{\text{ANC}}$ represents encoding coefficients from incoming edges $e', e'', e'''$ to the outgoing edge $e_2$, and the second column represents encoding coefficients from incoming edges $e', e'', e'''$ to the outgoing edge $e'''$. Using a matrix notation, the output data symbols from $v$
are $(y_{e_2,t} \quad y_{e''',t}) = (y_{e',t} \quad y_{e'',t} \quad
y_{e''',t-1})K_{\text{ANC}} $, i.e.,
\begin{align}
    y_{e''',t} & = y_{e',t}k_{e',e'''} +y_{e'',t}k_{e'',e'''}+y_{e''',t}0\notag \\
               & = y_{e',t}k_{e',e_1,1} +y_{e'',t}k_{e'',e_1,1}\notag \\
    y_{e_2,t}& = y_{e',t}k_{e',e_2}+y_{e'',t}k_{e'',e_2}+y_{e''',t-1}k_{e''',e_2} \notag \\
             & = y_{e',t}k_{e',e_1,0}+y_{e'',t}k_{e'',e_1,0}+y_{e''',t-1} \notag \\
    & = \sum_{\dot{e}\in\{e',e''\}} y_{\dot{e},t}k_{\dot{e},e_1,0}+y_{\dot{e},t-1}k_{\dot{e},e_1,1} \label{ye2tANC}
\end{align}
Clearly $y_{e_1,t}$ is equal to $y_{e_2,t}$. ARCNC therefore falls into the framework given by Ho et~al. \cite{HMK2006}, in the sense that the convolution process either
arises naturally from cycles with delays, or can be considered as
computed over self-loops appended to acyclic networks. Applying the
analysis from \cite{HMK2006}, we have the following theorem,

\begin{theorem}\label{thm:prob}
For multicast over a general network with $d$ sinks, the ARCNC algorithm over $\mathbb{F}_q$
can achieve a success probability of at least $(1-d/q^{t+1})^{\eta}$ at time $t$, if $q^{t+1}>d$, and $\eta$ is the number of links with
random coefficients.
\end{theorem}

\begin{proof} At node $v$, $k_{e',e}(z)$ at time $t$ is a polynomial with maximal degree $t$, i.e.,  $k_{e',e}(z)=k_{e',e,0}+k_{e',e,1}z+\cdots+k_{e',e,t}z^{t}$,  $k_{e',e,i}$ is randomly chosen over $\mathbb{F}_q$. If we group the  coefficients, the vector $k_{e',e}=\{k_{e',e,0},k_{e',e,1},\cdots,k_{e',e,t}\}$ is of
length $t+1$, and corresponds to a random element over the extension field
 $\mathbb{F}_{q^{t+1}}$. Using the result in \cite{HMK2006}, we conclude that the success probability of ARCNC at time $t$ is at least $(1-d/q^{t+1})^{\eta}$, as long as $q^{t+1}>d$.
\end{proof}

We could similarly consider the analysis done by Balli et al. \cite{BYZ2009}, which states that the success probability is at least
$(1-d/(q-1))^{|J|+1}$, $|J|$ being the number of encoding nodes, to show that a tighter lower bound can be given on the success probability
of ARCNC, when $q^{t+1}>d$.

\subsection{First decoding time}\label{sec:stoppingTime}
As discussed in Section~\ref{subsubsec:decoding}, we define the \emph{first decoding time} $T_r$ for sink $r$, $1\leq r\leq d$, as the time it takes $r$ to achieve decodability for the first time. We had called this variable the \emph{stopping time} in \cite{guo2011localized}. Also recall that when all sinks are able to decode, at each sink $r$, $T_r$ can be smaller than $L_r$, the degree of the global encoding kernel matrix $F_r(z)$. Denote by $T_N$ the time it takes for all sinks in the network to successfully decode, i.e., $T_N=\max\{T_1,\ldots,T_d\}$, then $T_N$ is also equal to $\max\{L_1,\ldots,L_d\}$. The following corollary holds:

\begin{corollary}
For any given $0<\varepsilon<1$, there exists a $T_0>0$ such that for any $t \geq T_0$, ARCNC solves the multicast problem with probability
at least $1-\varepsilon$, i.e., $P(T_N > t)  < \varepsilon$.\label{crlry2}
\end{corollary}
\proof Let $T_0  = \left\lceil \log_q d-\log_q(1-\sqrt[\eta]{1-\epsilon}) \right\rceil -1 $, then $T_0 + 1 \geq \lceil \log_q d \rceil$ since
$0 < \varepsilon <1$, and $(1-d/q^{T_0+1})^\eta>1-\varepsilon$. Applying Theorem~\ref{thm:prob} gives $P(T_N > t) \leq P(T_N >T_0) <
1-(1-d/q^{t+1})^{\eta}< \varepsilon$ for any $t\geq T_0$,
\endproof

Since $ Pr\{ \cup_{i=t}^{\infty}[T_N \le t]\}=1-Pr\{\cap
_{i=t}^{\infty}[T_N>t]\} 
>1-\varepsilon$,
Corollary~\ref{crlry2} shows that ARCNC converges and stops in a
finite amount of time with probability 1 for a multicast
connection.

Another relevant measure of the performance of ARCNC is the
\emph{average first decoding time},
$T_{\text{avg}}=\frac{1}{d}\sum_{r=1}^{d}T_r$. Observe that
$E[T_{\text{avg}}] \leq E[T_N]$, where
\begin{align*}
  E[T_N] 
   & = \sum_{t=1}^{\lceil \log_q d\rceil -1}P(T_N\geq t) + \sum_{t= \lceil \log_q d\rceil}^{\infty}{P(T_N\geq t)}\\
         & \le \lceil \log_q d\rceil -1 + \sum_{t=\lceil \log_q d\rceil}^{\infty}[1-(1-\frac{d}{q^t})^{\eta}] \\
         & = \lceil \log_q d\rceil -1 + \sum_{k=1}^{\eta}(-1)^{k-1}{\eta \choose                                                         k} \frac{d^k}{q^{\lceil \log_q d\rceil k}-1}\,.
\end{align*}
When $q$ is large, the summation term approximates $1-(1-d/q)^\eta$ by the
binomial expansion.  Hence as $q$ increases, the second term above
decreases to $0$, while the first term $\lceil \log_q d\rceil -1$ is
0. $E[T_{\text{avg}}]$ is therefore upper-bounded by a term
converging to 0; it is also lower bounded by 0 because at least one
round of random coding is required. Therefore, $E[T_{\text{avg}}]$
converges to $0$ as $q$ increases.  In other words, if the field
size is large enough, ARCNC reduces in effect to RLNC.

Intuitively, the average first decoding time of ARCNC depends on the network topology. In RLNC, all nodes are required in code in finite fields of the same size; thus the effective field size is determined by the worst case sink. This scenario corresponds to having all nodes stop at $T_N$ in ARCNC. ARCNC enables each node to decide locally what is a good constraint length to use, depending on side information from downstream nodes. Since $E[T_{\text{avg}}]\leq E[T_N]$, some nodes may be able to decode before $T_N$. The corresponding effective field size is therefore expected to be smaller than in RLNC. Two
possible consequences of a smaller effective field size are reduced decoding delay, and reduced memory requirements. In Section~\ref{sec:examples}, we confirm through simulations that such gains can be attained by ARCNC.

\subsection{Memory}\label{sec:memory}

To measure the amount of memory required by ARCNC, first recall
from Section~\ref{sec:basicDefs} that at each node $v$, the global
encoding kernel matrix $F_v(z)$, the local encoding kernel matrix
$K_v(z)$, and past data $y_{e'}(z)$ on incoming arcs $e'\in In(v)$
need to be stored in memory. $K_v(z)$ and $y_{e'}(z)$ should always
be saved because together they generate new data symbols to transmit
(see Eqs.~\eqref{eq:yet_k} and \eqref{eq:yet}). $F_v(z)$ should be
saved during the code construction process at intermediate nodes,
and always at sinks, since they are needed for decoding.

Let us consider individually the three contributors to memory use.
Firstly, recall from Section~\ref{sec:basicDefs} that $F_v(z)$ can
be viewed as a polynomial in $z$. When all sinks are able to decode,
at node $v$, $F_v(z)$ has degree $L_v$, with
coefficients from $\mathbb{F}_q^{m\times In(v)}$. The total amount
of memory needed for $F_v(z)$ is therefore proportional to
$\lceil\log_2 q\rceil m In(v) (L_v+1)$. Secondly, from
Eq.~\eqref{eq:fet}, we see that the length of a local encoding
kernel polynomial $k_{e',e}(z)$ should be equal to or smaller than
that of $f_e(z)$. Thus, the length of $K_v(z)$ should also be equal
to or smaller than that of $F_v(z)$. The coefficients of $K_v(z)$
are elements of $\mathbb{F}_q^{In(v)\times Out(v)}$. Hence, the
amount of memory needed for $K_v(z)$ is proportional to
$\lceil\log_2 q\rceil Out(v) In(v) (L_v+1)$. Lastly, a direct
comparison between Eqs.~\eqref{eq:fet} and \eqref{eq:yet} shows that
memory needed for $y_{e'}(z)$, $e'\in In(v)$ is the same for that
needed for $F_v(z)$. In practical uses of network coding, data can
be transmitted in packets, where symbols are concatenated and
operated upon in parallel. Packets can be very long in length.
Nonetheless, the exact packet size is irrelevant for
comparing memory use between different network codes, since all
comparisons are naturally normalized to packet lengths.

Observe that, $m$ is the number of symbols in the source message,
determined by the min-cut of the multicast connection, independent
of the network code used. Similarly, $In(v)$ and $Out(v)$ are
attributes inherent to the network topology. To compare the memory
use of different network codes, we can omit these terms, and define
the average memory use of ARCNC by the following common factor:
\begin{align}
W_{\text{avg}} \triangleq \frac{\lceil\log_2 q\rceil}{|\mathcal{V}|}\sum_{v\in\mathcal{V}}(L_v+1)\,.\label{eq:Wavg}
\end{align}
\noindent In RLNC, $L_v=0$, and the expression simplifies to $\lceil\log_2 q\rceil$, which is the amount of memory needed for a single finite field element.

One point to keep in mind when measuring memory use is that even after a sink achieves decodability, its code length can still increase, as long as at least one of its ancestors has not stopped increasing code length. We say a non-source node $v$ is \emph{related} to a sink $r$ if $v$ is an ancestor of $r$, or if $v$ shares an ancestor, other than the source, with $r$. Hence, $L_r$ is dependent on all nodes related to $r$.
\subsection{Complexity}
To study the computation complexity of ARCNC, first observe
that, once the adaptation process terminates, the computation needed
for the ensuing code is no more than a regular CNC. In fact, the
expected computation complexity is proportional to the average code
length of ARCNC. We therefore omit the details of the complexity
analysis of regular CNC here and refer interested readers to
\cite{EF2004}.

For the adaptation process, the encoding operations are described by
Eq.~\eqref{eq:fet}. If the algorithm stops at time $T_N$, the
number of operations in the encoding steps is
$O(D_{in}|\mathcal{E}|T_N^2m)$, where
$D_{in}=\max_{v\in\mathcal{V}}|In(v)|$.

To determine decodability at a sink $r$, we check if the rank of $F_r(z)$ is $m$. A straight-forward approach is to check whether its determinant is a non-zero
polynomial. Alternatively, Gaussian elimination could be
applied. At time $t$, because $F_r(z)$ is an $m \times |In(r)|$
matrix and each entry is a polynomial with degree $t$, the
complexity of checking whether $F_r(z)$ is full rank is
$O(D_{in}^22^mmt^2)$. Instead of computing the determinant or using
Gaussian elimination directly, we propose to check the conditions
given in Section~\ref{sec:algAcyclic}. For each sink $r$, at time
$t$, determining $rank\left({\begin{array}{*{20}c} F_0 & F_1 &
\cdots F_{t} \end{array}}\right)$ requires $O(D_{in}^2mt^2)$ operations. If the first test passes, we calculate $rank(M_{t})$ and $rank(M_{t-1})$ next. Observe that $rank(M_{t-1})$ was computed during the last iteration. $M_t$ is a
$(t+1)m \times (t+1)|In(r)|$ matrix over field $\mathbb{F}_q$. The
complexity of calculating $rank(M_t)$ by Gaussian elimination is
$O(D_{in}^2mt^3)$. The process of checking decodability is performed
during the adaptation process only, hence the computation complexity
here can be amortized over time after the coding coefficients are
determined. In addition, as decoding occurs symbol-by-symbol, the
adaptation process itself does not impose any additional delays.

\section{Examples}\label{sec:examples}
In this section, we describe the application of ARCNC in three structured networks: the combination and sparsified combination networks, which are acyclic, and the shuttle network, which is cyclic. For the combination network, we bound the expected average first decoding time; for the sparsified combination network, we bound the expected average memory requirement. In addition, the shuttle network is given as a very simple
example to illustrate how ARCNC can be applied in cyclic networks.

\subsection{Combination Network} \label{example:combination}
A $n\choose m$ combination network contains a single source $s$ that multicasts $m$ independent messages over $\mathbb{F}_q$ through $n$ intermediate nodes to $d$ sinks \cite{NWT2005}; each sink is connected to a distinct set of $m$ intermediate nodes, and $d={n\choose m}$. Fig.~\ref{fig:combination} illustrates the topology of a combination network. Assuming unit capacity links, the
min-cut to each sink is $m$. It can be shown that, in combination networks, routing is insufficient and network coding is needed to achieve the multicast capacity $m$. Here coding is performed only at $s$, since each intermediate node has only $s$ as a parent node; an intermediate node simply relays to its children data from $s$. For a general $n \choose m$ combination
network, we showed in \cite{guo2011localized} that the expected
average first decoding time can be significantly improved by ARCNC
when compared to the deterministic BNC algorithm. We restate the
results here, with details of the derivations included.

\begin{figure}[t!]
  \centering
  \includegraphics[width=5.5cm]{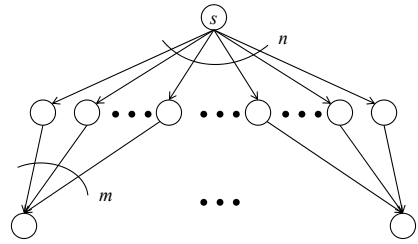}\\
  \caption{A combination network}\label{fig:combination}
\end{figure}

At time $t-1$, for a sink $r$ that has not satisfied the
decodability conditions, $F_r(z)$ is a size $m\times m$ matrix of
polynomials of degree $t-1$. $F_r(z)$ has full rank with probability
\begin{align}
Q&=(q^{tm}-1)(q^{tm}-q^t)\cdots(q^{tm}-q^{t(m-1)})/q^{tm^2} \notag\\
 &=(1-\frac{1}{q^{tm}})(1-\frac{1}{q^{t(m-1)}})\cdots(1-\frac{1}{q^t})\notag\\
 &=\prod_{l=1}^m\left(1-\frac{1}{q^{tl}}\right)\,. \label{eq:11suc}
\end{align}
Hence, the probability that sink $r$ decodes after time $t-1$ is
\begin{align}
P(T_r\geq t) & = 1- Q = 1- \prod_{l=1}^m\left(1-\frac{1}{q^{tl}}\right) \,, \quad t\geq 0. \label{eq:pt}
\end{align}
The expected first decoding time for sink node $r$ is therefore upper and lower-bounded as follows.
\begin{align}
E[T_r] & = \sum_{t=1}^\infty tP(T_r = t) = \sum_{t=1}^\infty P(T_r\geq t) \notag\\
     & = \sum_{t=1}^\infty \left(1-\prod_{i=1}^m\left(1-\frac{1}{q^{tr}}\right)\right)\\
     & < \sum_{t=1}^\infty \left(1-\left(1-\frac{1}{q^{t}}\right)^m\right)\label{eq:13mean}\\
     & = \sum_{t=1}^\infty \left(1-\sum_{k=0}^m(-1)^k{m\choose k}\left(\frac{1}{q^t}\right)^k\right)\label{eq:14mean}\\
     & = \sum_{k=1}^m(-1)^{k-1}{m\choose k}\left(\sum_{t=1}^{\infty}\frac{1}{q^{tk}}\right)\\
     & = \sum_{k=1}^m (-1)^{k-1}{m\choose k}\frac{1}{q^k-1} \triangleq ET_{UB}(m,q) \,. \label{eq:ETUB}
\end{align}
\begin{align}
E[T_r] & =  \sum_{t=1}^\infty\left(1-\prod_{l=1}^m\left(1-\frac{1}{q^{tl}}\right)\right) \\
     &  > \sum_{t=1}^\infty \left(1-\left(1-\frac{1}{q^{tm}}\right)^m\right)\\
     & = \sum_{k=1}^m (-1)^{k-1}{m\choose k}\frac{1}{q^{km}-1} \triangleq
     ET_{LB}(m,q)\,. \label{eq:ETLB}
\end{align}
Recall from Section~\ref{sec:stoppingTime} that the expected average
first decoding time is $E[T_{\text{avg}}] = E\left[
\frac{1}{d}\sum_{r=1}^{d} T_r \right]$. In a combination network, $
E[T_{\text{avg}}]$ is equal to $E[T_r]$. Consequently,
$E[T_{\text{avg}}]$ is upper-bounded by $ET_{UB}$, defined by
Eq.~\eqref{eq:ETUB}. $ET_{UB}$ is a function of $m$
and $q$ only, independent of $n$. For example, if $m=2$, $q=2$,
$ET_{UB} = \frac53$. If $m$ is fixed, but $n$ increases, $
E[T_{\text{avg}}]$ does not change. In addition, if $q$ is large,
$ET_{UB}$ becomes 0, consistent with the general analysis in
\cite{guo2011localized}.

Next, we want to bound the variance of $T_{\text{avg}}$, i.e.,
\begin{align}
\hspace{-5pt} var[T_{\text{avg}}]   & = E[T_{\text{avg}}^2] - E^2[T_{\text{avg}}] \notag \\[4pt]
         & = E\left[\left(\frac{1}{d}\sum_{r=1}^d T_r\right)^2\right] - E^2[T_r] \notag\\[4pt]
         & = \frac{E[T_r^2]}{d} + \left(\sum_{r=1}^{d}\sum_{r\neq r'}\frac{E(T_rT_{r'})}{d^2}\right) - E^2[T_r]\,. \label{eq:varTavg}
\end{align}
We upper-bound the terms above one by one. First,
\begin{align}
E[T_r^2] & = \sum_{t=1}^\infty t^2P(T_r=t) \\
         & = \sum_{t-1}^\infty t^2(P(T_r \ge t)-P(T_r \ge t+1)) \\
         & = \sum_{t=1}^\infty ((t+1)^2-t^2)P(T_r\geq t) \label{eq:lowdim}\\
         & < \sum_{t=1}^\infty (2t+1) \left(1-\left(1-\frac{1}{q^{t}}\right)^m\right) \label{eq:upbound11}\\
         & < ET_{UB} + 2 \sum_{k=1}^m (-1)^{k-1}{m\choose k}\sum_{t=1}^\infty \frac{t}{q^{tk}}\label{eq:apple}\\
         & = ET_{UB}  + 2 \sum_{k=1}^m (-1)^{k-1}{m\choose k}\left(\frac{q^k}{q^k-1}\right)^2 \label{eq:orange}\\
         & \triangleq (ET^2)_{UB}
\end{align}
Eq.~\eqref{eq:lowdim} follows through organization and simplifying.
Eq.~\eqref{eq:upbound11} is obtained by replacing the terms in
Eq.~\eqref{eq:lowdim} with the upperbound of Eq.~\eqref{eq:11suc}.
We represent Eq.~\eqref{eq:upbound11} with binomial expansion and
substitute
with the upperbound in Eq.~\eqref{eq:ETUB}. 
Next, let
$\rho_\lambda=E[T_rT_{r'}]$ if sinks $r$ and $r'$ share $\lambda$
parents, $0 \leq \lambda < m$. Thus, $\rho_0 = E^2[T_r]$. When
$\lambda \neq 0$, given sink $r$ succeeds in decoding at time $t_1$,
the probability that sink $r'$ has full rank before $t_2$ is
lower-bounded as follows,
\begin{align}
P(T_{r'} < t_2|T_r=t_1) > \prod_{l=1}^{m-\lambda}\left(1-\frac{1}{q^{t_2l}}\right) >\left(1-\frac{1}{q^{t_2}}\right)^{m-\lambda}\,.
\end{align}
Consequently, if $\lambda\neq0$,
\begin{align}
 \rho_\lambda & = E[T_rT_{r'}] \\
            & = \sum_{t_1=1}^\infty \sum_{t_2=t_1}^\infty t_1 t_2 P(T_r = t_1)P(T_{r'}=t_2|T_r=t_1)\\
            & = \sum_{t_1=1}^\infty t_1 P(T_r = t_1) \sum_{t_2=1}^\infty P(T_{r'}\geq t_2|T_r=t_1) \\
            & < \sum_{t_1=1}^\infty t_1 P(T_r = t_1) \sum_{t_2=1}^\infty \left(1-\left(1-\frac{1}{q^{t_2}}\right)^{m-\lambda}\right)\\
            & < \sum_{t_1=1}^\infty t_1 P(T_r = t_1) \sum_{k=1}^{m-\lambda}(-1)^{k-1}{ m-\lambda \choose k } \frac{1}{q^k-1}\\
            & < ET_{UB}\left(\sum_{k=1}^{m-\lambda}(-1)^{k-1}{m-\lambda\choose k}\frac{1}{q^k-1}\right)\\
            & \triangleq \rho_{\lambda,UB}
\end{align}

\noindent Let $\rho_{UB} =
\max\{\rho_{1,UB},\ldots,\rho_{m-1,UB}\}$. For a sink $r$, Let the
number of sinks  that share at least one parent with $r$ be
$\Delta$, then $\Delta = d-1-{n-m\choose m}$. Thus, the middle term
in Eq.~\eqref{eq:varTavg} is bounded by
$\frac{\Delta}{d}\rho_{UB}+\frac{d-1-\Delta}{d}E^2[T_r]$ and
\begin{align}
var[T_{\text{avg}}]
          & < \frac{(ET^2)_{UB}}{d} + \frac{\Delta}{d}\rho_{UB} - \left(\frac{\Delta+1}{d}\right)ET^2_{LB} \label{eq:varTvalue}\,.
\end{align}
Depending on the relative values of $n$ and $m$, we have the following three cases.

\begin{itemize}
 \item $n>2m$, then ${n-m \choose m} = \frac{(n-m)!}{m!(n-2m)!}$,and
\begin{align}
\frac{\Delta}{d}  & = 1-\frac1d -\frac{{n-m \choose m}}{d}  \\
     & = 1-\frac1d -\frac{(n-m)!(n-m)!}{n!(n-2m)!} \\
        & = 1-\frac1d-\frac{(n-m)(n-m-1)\ldots(n-2m+1)}{n(n-1)\ldots(n-m+1)}\\
        & = 1-\frac1d-\left(\frac{n-m}{n}\right)\ldots \left(\frac{n-2m+1}{n-m+1}\right)\\
        & < 1-\frac1d-\left(\frac{n-2m+1}{n-m+1}\right)^m \label{eq:Delta_d}\,.
\end{align}
Observe from Eqs.~\eqref{eq:ETUB} and \eqref{eq:ETLB} that all of
the upper-bound and lower-bound constants are functions of $m$ and
$q$ only. If $m$ and $q$ are fixed and $n$ increases, in
Eq.~\eqref{eq:Delta_d}, both $\frac{\Delta}{d}$ and
$\frac{\Delta+1}{d}$ approaches 0.  Therefore, $var(T)$ diminishes
to 0. Combining this result with the upper-bound $ET_{UB}$, we can
conclude that, when $m$ is fixed, even if more intermediate nodes
are added, a large proportion of the sink nodes can still be decoded
within a small number of coding rounds.

\item $n=2m$, then ${n-m \choose m} = 1$, $\frac{\Delta}{d} = 1-\frac{2}{d}$, and
\begin{align}
\hspace{-10pt} var[T_{\text{avg}}]
& <\frac{(ET^2)_{UB}}{d} + \left(1-\frac{2}{d}\right)\rho_{UB}- \left(1-\frac{1}{d}\right)ET_{LB}^2 \notag\\
& <\frac{(ET^2)_{UB}}{d} + \rho_{UB}- \left(1-\frac{1}{d}\right)ET_{LB}^2
\end{align}
Here $m$ and $n$ are comparable in scale, and the bounds depend on the exact values of $ET^2_{UB}$, $\rho_{UB}$ and $ET_{UB}$. We will illustrate through simulation in Section~\ref{subsec:simuCombination} that in this case, $T_{avg}$
also converges to 0.

\item $n<2m$, then ${n-m \choose m} = 0$, $\frac{\Delta}{d} = 1-\frac{1}{d}$, and
\begin{align}
var[T_{\text{avg}}] & < \frac{(ET^2)_{UB}}{d} + \rho_{UB}- ET_{LB}^2,
\end{align}
\noindent similar to the second case above.
\end{itemize}

Comparing with the deterministic BNC by Xiao et~al. \cite{xiao2008binary}, we can see that, for a large combination network, with fixed $q$ and $m$, ARCNC achieves much lower first decoding time. In BNC, the block length is required to be $p\geq n-m$ at minimum; the decoding delay increases at least linearly with $n$, where as in ARCNC, the expected average first decoding time is independent of the value of $n$. On the other hand, with RLNC \cite{HMK2006}, the multicast capacity can be achieved with probability $(1-d/q)^n$.  The exponent $n$ is the number of links with random coefficients; since each intermediate node has the
source as a single parent, coding is performed at the source only, and coded data are transmitted on the $n$ outgoing arcs from the source. When $q$ and $m$ are fixed, the success probability of RLNC decreases exponentially in $n$. Thus, an exponential number of trials is needed to find a successful RLNC. Equivalently, RLNC can use an increasingly large field size $q$ to maintain the same decoding probability.

So far we have used $n\choose m$ combination networks explicitly to illustrate the operations and the decoding delay gains of ARCNC. It is important to note, however, that this is a very restricted family of networks, in which only the source is required to code, and each sink shares at least $1$ parent with
other ${n \choose m}-{n-m\choose m}-1$ sinks. In terms of memory, if sink $r$ cannot decode, all sinks related to $r$ are required to increase their memory capacity. Recall from Subsection~\ref{sec:memory} that a non-source node $v$ is said to be related to a sink $r$ if $v$ is an ancestor of $r$, or if $v$ shares an non-source ancestor with $r$. As $n$ becomes larger, the number of nodes related to $r$ increases, especially if $m$ increases too. Thus, in combination networks, we do not see considerable gains in terms of memory overheads when compared with BNC, unless $m$ is small. In more general networks, however, when sinks do not share ancestors with as many other sinks, ARCNC can achieve gains in terms of memory overheads as well, in addition to decoding delay. As an example, we define a \emph{sparsified combination network} next.

\subsection{Regular sparsified Combination Network} \label{example:sparsecombnet}

\begin{figure}[t!]
  \centering
  \includegraphics[width=6cm]{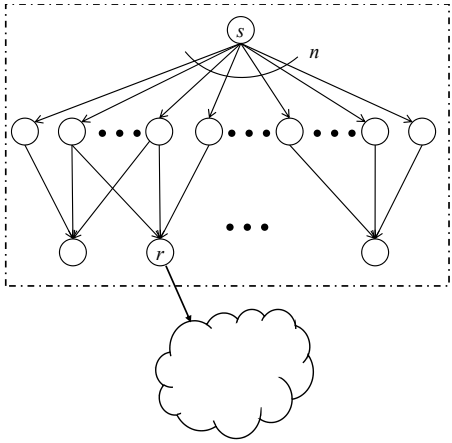}\\
  \caption{A regular sparsified combination network (inside the dotted frame) with an extension. }\label{fig:dconstraint}
\end{figure}

We define a regular sparsified combination network as a modified combination
network, with only \emph{consecutive} intermediates nodes connected to unique sink nodes. The framed component in Fig.~\ref{fig:dconstraint} illustrates its structure. Source $s$ multicasts $m$ independent messages through $m$ intermediate nodes to each sink, with $n-m+1$ sinks in total. This topology can be viewed as an abstraction of a content distribution network, where the source distributes data to intermediate servers, and clients are required to connect to $l$ servers closest in distance to collect enough degrees of freedom to obtain the original data content. This network can be arbitrarily large in scale.

In a regular sparsified combination network, the number of other sinks related to a sink $r$ is fixed at $2(m-1)$, and is even smaller if $r$'s parents are on the edge of the intermediate layer. Thus, the average first decoding time of sinks in a regular sparsified combination network behaves similarly to the fixed $m$ case discussed in the previous subsection, approaching 0 as $n$ goes to infinity.

On other hand, since now each intermediate node is connected to a fixed number of $m$ sinks as well, when a sink $r$ fails to decode and requests an increment in code length, a maximum of $m+2(m-1)=3m-2$ related nodes are required to increase their memory capacity. To compute $W_\text{avg}$ using Eq.~\eqref{eq:Wavg}, observe that for a sink $r$, assuming there are the maximum number of $2(m-1)$ other sinks related to $r$, the cumulative probability distribution of $L_r$ is as follows
\begin{align*}
& \Pr\{L_r<t\} \\ & = \Pr \{T_{r-m+1} <t,\ldots, T_{r} <t, \ldots, T_{r+m-1} <t\}\\
& = \Pr \{T_{r-m+1} <t\} \Pr \{T_{r-m+2} <t|T_{r-m+1} <t\}  \\
& \quad \ldots \Pr \{T_{r+m-1} <t|T_{r-m+1} <t,\ldots, T_{r+m-2} <t\} \\
& = Q \left(1-\frac{1}{q^t}\right)^{2m-2}
\end{align*}
\noindent where $Q$ is defined in Eq.~\eqref{eq:11suc}. Thus, using the derivation from Eq.~\eqref{eq:pt} to \eqref{eq:ETUB}, we have
\begin{align*}
E[L_r] & = \sum_{t=1}^\infty P(L_r\geq t) \notag\\
& = \sum_{t=1}^\infty \left(1-\prod_{i=1}^m\left(1-\frac{1}{q^{tr}}\right)\left(1-\frac{1}{q^{t}}\right)^{2m-2}\right)\\[4pt]
& < ET_{UB}(3m-2, q)
\end{align*}

\noindent Similarly, for an intermediate node $v$, we can bound $E[L_v]$ by $
ET_{UB}(2m, q)$. Clearly $W_\text{avg,ARCNC}$ computed using
Eq.~\eqref{eq:Wavg} is a function of $m$ and $q$ only, independent
of $n$. In other words, in a regular sparsified combination network,
since each sink has a fixed number of parents, and are related to a
fixed number of other sinks through its parents, the average amount
of memory use across the network is independent of $n$.

On the other hand, assume a field size of $q_R$ is used for RLNC code generation. There is a single coding node in the network, with $n-m+1$ sinks. To guarantee an overall success probability larger than $1 - \varepsilon $, we have $(1 - \frac{n-m+1}{q_R-1
})^2  > 1 - \varepsilon $ from \cite{BYZ2009}. Hence
\begin{align*}
E[W_{\text{avg,RLNC}} ] = \left\lceil {\log _2 q_R } \right\rceil  > \left\lceil \log _2 (1 +
\frac{n-m+1}{1-\sqrt{1 - \varepsilon }})\right\rceil,
\end{align*}
which can be very large if $n$ is large and $\varepsilon$ is small. 

Comparing the lower-bound on $E[W_{\text{avg},\text{RLNC}}]$ and the
upper-bound on $E[W_{\text{avg},\text{ARCNC}}]$, we see that the
gain of ARCNC over RLNC in terms of memory use is infinite as $n$
increases because $E[W_{\text{avg},\text{ARCNC}}]$ is bounded by a
constant value.
 
An intuitive generalization of this observation is to extend
this regular sparsified combination network by attaching another
arbitrary network off one of the sinks, as shown in
Fig.~\ref{fig:dconstraint}. Regardless of the depth of this
extension from the sink $r$, as $n$ increases, memory overheads can
be significantly reduced with ARCNC when compared with RLNC, since
most of the sinks and intermediate nodes are unrelated to $r$, thus
not affected by the decodability of sinks within the extension.

\subsection{Shuttle Network}\label{example:shuttle}
\begin{figure}[t!]
  \centering
  \includegraphics[width=5.7cm]{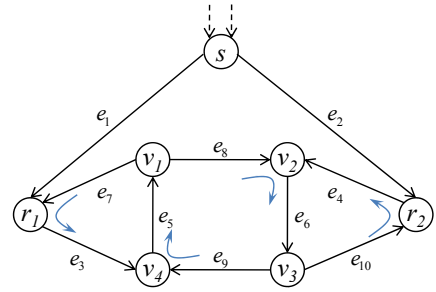}\\
  \caption{The shuttle network. Each link has unit capacity. $s$ is the source; $r_1$ and $r_2$ are sinks each with a min-cut of 2. Edges are directed and labeled as $e_i$, $1\leq i\leq 10$. Edges indices are assigned according to Section~\ref{sec:algCyclic}. An adjacent pair $(e',e)$ is labeled with a curved pointer if $e'\succeq e$.}\label{fig:shuttlenet}
\end{figure}
In this section, we illustrate the use of ARCNC in cyclic networks
by applying it to a shuttle network, shown in
Fig.~\ref{fig:shuttlenet}. We do not provide a formal definition for
this network, since its topology is given explicitly by the figure.
Source $s$ multicasts to sinks $r_1$ and $r_2$. Edges $e_i$, $1\leq
i\leq 10$, are directed. The edge indices have been assigned
according to  Section~\ref{sec:algCyclic}. An adjacent pair $(e',e)$
is labeled with a curved pointer if $e'\succeq e$. There are three
cycles in the network; the left cycle is formed by $e_3$, $e_5$, and
$e_7$; the middle cycle is formed by $e_5$, $e_8$, $e_6$, and $e_9$;
the right cycle is formed by $e_4$, $e_6$, and $e_{10}$. In this
example, we use a field size $q=2$. At node $v$, the local encoding
kernel matrix is $K_v(z)=(k_{e',e}(z))_{e'\in In(v), e\in
Out(v)}=K_{v,0}+K_{v,1}z+K_{v,2}z^2+\ldots$; each local encoding
kernel is a polynomial, $k_{e',e}(z)
=k_{e',e,0}+k_{e',e,1}z+k_{e',e,2}z^2+\ldots$. At $s$, assume
$f_{e_1}(z) = {1 \choose 0}$, and $f_{e_2}(z) = {0 \choose 1}$,
i.e., the data symbols sent out from $s$ at time $t$ are $y_{e_1,t}
= x_{1,t}$ and $y_{e_2,t} = x_{2,t}$, respectively. The source can
also linearly combine source symbols before transmitting on outgoing
edges.

At $t=0$, we assign 0 to local encoding kernel coefficients
$k_{e',e,0}$ if $e'\succeq e$; and choose  $k_{e',e,0}$ uniformly
randomly from $\mathbb{F}_2$ otherwise. One possible assignment is
given in Fig.~\ref{fig:shuttle0}. Here we circle $k_{e',e,0}$ if
$e'\succeq e$. Since $q=2$, we set all other local encoding kernel
coefficients to 1. The data messages transmitted on each edge at
$t=0$ are then derived and labeled on the edge. Observe that, $r_1$
receives $x_{1,0}$ and $r_2$ receives $x_{2,0}$; neither is able to
decode both source symbols. Hence no acknowledgment is sent in the
network.

\begin{figure}[t!]
  \centering
  \includegraphics[width=8.5cm]{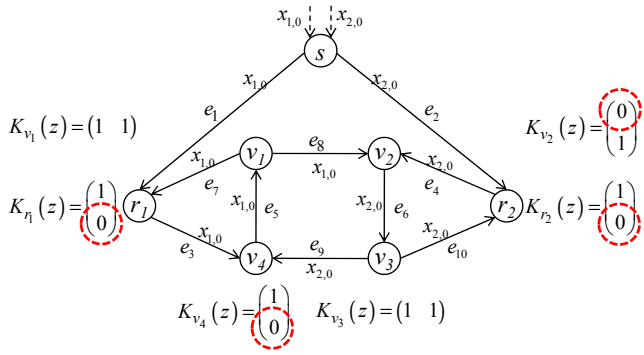}\\
  \caption{An example of local encoding kernel matrices at $t=0$. For a node $v$, $K_v(z)=(k_{e',e}(z))_{e'\in In(v), e\in Out(v)}=K_{v,0}+K_{v,1}z+K_{v,2}z^2+\ldots$. For any adjacent pair $(e',e)$ where $e'\succeq e$, $k_{e',e,0}=0$. Each edge $e$ is labeled with the data symbol $y_{e,t}$ it carries, e.g., $y_{e_1,0}=x_{1,0}$.}\label{fig:shuttle0}
\end{figure}
\begin{figure}[t!]
  \centering
  \includegraphics[width=8.5cm]{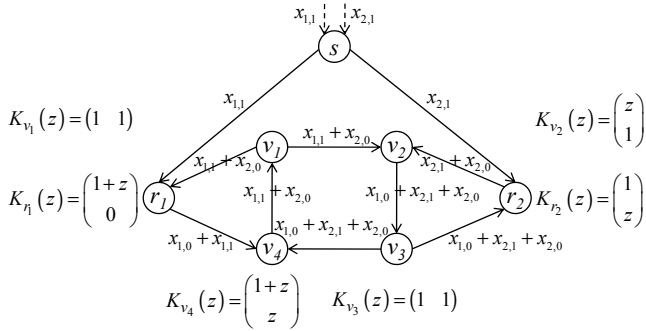}\\
  \caption{An example of local encoding kernel matrices at $t=1$.} 
\label{fig:shuttle1}
\end{figure}

At $t=1$, we proceed as in the acyclic case, randomly choosing coefficients $k_{e',e,1}$ from $\mathbb{F}_2$. Since no acknowledgment has been sent by $r_1$ or $r_2$, all local encoding kernels increase in length by 1. One possible coding kernel coefficient assignment is given in Fig.~\ref{fig:shuttle1}. Both $v_1$ and $v_3$ have one incoming edge only and thus route instead of code, i.e., $K_{v_1}(z)=K_{v_3}(z)=(1 \quad 1)$. The other local encoding kernel matrices in this example are as follows

\begin{align*}
K_{r_1,1}(z) & =
  {k_{e_1,e_3,0} \choose k_{e_7,e_3,0}}
   +  {k_{e_1,e_3,1} \choose k_{e_7,e_3,1}} z
  = {1\choose 0} + {1\choose 0}z\,,\\[2pt]
 K_{r_2,1}(z) & =  {k_{e_2,e_4,0} \choose k_{e_{10},e_4,0}}
   +  {k_{e_2,e_4,1} \choose k_{e_{10},e_4,1}} z
  = {1 \choose 0} + {0 \choose 1}z\,,  \\[2pt]
K_{v_4,1}(z)  & =  {k_{e_3,e_5,0} \choose k_{e_9,e_5,0}}
   +  {k_{e_3,e_5,1} \choose k_{e_9,e_5,1}} z
  = {1 \choose 0}+{1 \choose 1}z\,,\\[2pt]
 K_{v_2,1}(z) & =  {k_{e_4,e_6,0} \choose k_{8_1,e_6,0}}
   +  {k_{e_4,e_6,1} \choose k_{e_8,e_6,1}} z
  = {0 \choose 1}+{1 \choose 0}z.
\end{align*}
Data symbols generated according to Eq.~\eqref{eq:yet} for this
particular code are also labeled on the edges. For
example, on edge $e_5 =(v_4,v_1)$, the data symbol transmitted at
$t=1$ is
\begin{align}
 y_{e_5,1} & =   y_{e_3,0}k_{e_3,e_5,1}+ y_{e_3,1}k_{e_3,e_5,0} \notag \\
               & \quad \quad + y_{e_9,0}k_{e_9,e_5,1}+ y_{e_9,1}k_{e_9,e_5,0} \label{eq:ye5} \\
           & =   x_{1,0}\cdot 1 + (x_{1,0}+x_{1,1})\cdot 1 + x_{2,0}\cdot 1 + y_{e_9,1}\cdot 0 \notag \\
           & = x_{1,1} + x_{2,0} \notag
\end{align}

Observe that there are no logical contradictions in any of the three cycles. For example, in the middle cycle, on $e_5$, regardless of the value of $y_{e_9,1}$, the incoming data symbol at $t=1$, $y_{e_5,1}$, can be evaluated as in Eq.~\eqref{eq:ye5}. In other words, in evaluating the global
encoding kernel coding coefficients according to Eqs.~\eqref{eq:fe}
and \eqref{eq:fet}, even though $f_{e_9,1}$ is unknown, $f_{e_5,1}$
can still be computed since $k_{e_9,e_5,0}=0$.

Also from Fig.~\ref{fig:shuttle1}, observe that both sinks can
decode two source symbols at $t=1$: $r_1$ can decode $x_{1,1}$ and
$x_{2,0}$, while $r_2$ can decode $x_{2,1}$ and $x_{1,0}$.
Equivalently, we can compute the global encoding kernel matrices and
check the decodability conditions given in
Section~\ref{sec:algAcyclic}. We omit the details here, but
interested readers can verify using Eq.~\eqref{eq:fe} that the
global encoding matrices are $F_{r_1}(z)={1\,\,\, 1 \choose 0\,\,\,
z}$, $F_{r_2}(z)={0\quad z \,\, \choose 1\,\,\, 1+z}$, and the
decodability conditions are indeed satisfied. Acknowledgments are
sent back by both sinks to their parents, code lengths stop to
increase, and ARCNC terminates. The first decoding time for both
sinks is therefore $T_{r_1}=T_{r_2}=1$.

As we have discussed in Section~\ref{sec:algCyclic}, the deterministic edge indexing scheme proposed is an universal but heuristic way of assigning local encoding kernel coefficients at $t=0$. In this shuttle network example, observe from Fig.~\ref{fig:shuttle0} that in the middle cycle composed of edges $e_8$, $e_6$, $e_9$ and $e_5$, this scheme introduces two zero coefficients, i.e., $k_{e_8,e_6,0}=0$, and $k_{e_9,e_5,0}=0$. A better code would be to allow one of these two coefficients to be non-zero. For example, if $k_{e_8,e_6,0}=1$, $K_{v_2}(z)={1 \choose 1}$ at $t=0$. It can be shown in this case that the data symbol transmitted on $e_{10}$ to $r_2$ is $y_{e_{10},0}=x_{1,0}+x_{2,0}$,
enabling $r_2$ to decode both source symbols at time $t=0$.

\section{Simulations}\label{sec:simulations}
We have shown analytically that ARCNC converges in finite steps with
probability 1, and that it can achieve gains in decoding time or memory in combination networks. In what follows, we want to verify these results through simulations, and to study numerically whether similar behaviors can be observed in random networks. We implemented the proposed encoder and decoder in \textsc{matlab}. In all instances, it can be observed that decoding success was achieved in a finite amount of time. All results plotted in this section are averaged over 1000 runs.

\subsection{Combination Network}\label{subsec:simuCombination}
Recall from Section~\ref{example:combination} that an upper bound $ET_{UB}$ and a lower bound $ET_{LB}$ for the average expected first decoding time $E[T_{\text{avg}}]$ can be computed for a $n \choose m$ combination network. Both are functions of $m$ and $q$, independent of $n$. In evaluating $var[T_{\text{avg}}]$, three cases were considered, $n>2m$, $n=2m$, and $n<2m$. When $n>2m$, the number of sinks unrelated to a given sink $r$ is significant. If it takes $r$ multiple time steps to achieve decodability, not all other sinks and intermediate nodes have to continue increasing their encoding kernel length to accommodate $r$. Thus, ARCNC can offer
gains in terms of decoding delay and memory use. We show simulation results below for the case when $m$ is fixed at the value of $2$, while $n$ increases. By comparison, if $n\leq2m$, there is a maximum of one sink related to a given sink $r$. We show simulation results below for the case of $n=2m$.

\begin{figure}[t!]
  \centering
  \includegraphics[width=8.5cm]{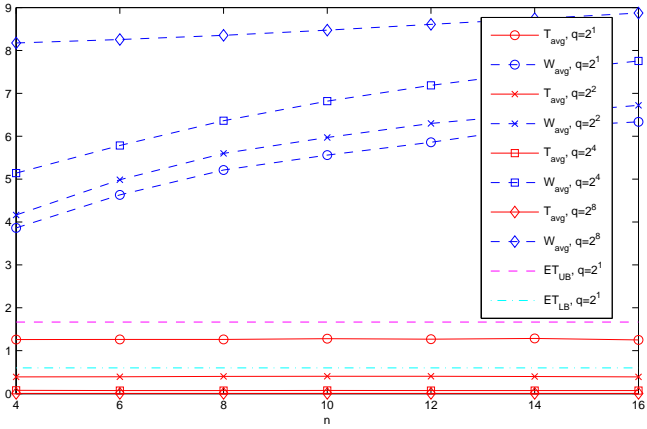}\\
  \caption{Average first decoding time and average memory use. $m=2$, $n$ increases, field size $q$ also increases. Also plotted are the computed upper and lower bounds on $T_{\text{avg}}$ for $q=2$.}
  \label{fig:combination_fixedM_2_qq}
\end{figure}

\subsubsection{$n>2m$, fixed $m$, $m=2$}
Fig.~\ref{fig:combination_fixedM_2_qq} plots the average first decoding time $T_{\text{avg}}$, corresponding upper and lower bounds $ET_{UB}$, $ET_{LB}$, and average memory use $W_{\text{avg}}$, as defined in Section~\ref{sec:memory}. Here $m$ is fixed to the value of 2, $n$ increases from 4 to 16, and the field size is $q=2$. As discussed in Section~\ref{example:combination}, $ET_{UB}$ and
$ET_{LB}$ are independent of $n$. As $n$ increases, observe that $T_{\text{avg}}$ stays approximately constant at about 1.3, while $W_{\text{avg}}$ increases sublinearly. When $n=16$, $W_{\text{avg}}$ is approximately 6.3. On the other hand, recall from \cite{BYZ2009} that a lower bound on the success probability of
RLNC is $(1-d/(q-1))^{|J|+1}$, where $|J|$ is the number of encoding nodes. In a combination network with $n=16$ and $m=2$, $|J|=1$ since only the source node codes. For a target decoding probability of $0.99$, we have $(1-{16\choose 2}/(q-1))^2 \geq 0.99$, thus $q>2.4\times 10^4$, and $\lceil \log_2q \rceil \geq 15$. Since each encoding kernel contains at least one term, $W_{\text{avg}}$ is lower bounded by $\lceil \log_2q\rceil$. Hence, using ARCNC here reduces memory use by half when compared with RLNC.

Fig.~\ref{fig:combination_fixedM_2_qq} also plots $T_{\text{avg}}$
and $W_{\text{avg}}$ when field size $q$ increases from $2^1$ to
$2^8$. As field size becomes larger, $T_{\text{avg}}$ approaches 0.
When $q=2^8$, the value of $T_{\text{avg}}$ is close to $0.004$. As
discussed in Section~\ref{example:combination}, when $q$
becomes sufficiently large, ARCNC terminates at $t=0$, and
generates the same code as RLNC. Also observe from this figure that
as $n$ increases from 4 to 16, $W_{\text{avg}}$ increases as well,
but at different rates for different field sizes. Again,
$W_{\text{avg}}$ is lower bounded by $\lceil \log_2q \rceil$. When
$q=2^8$, $W_{\text{avg}}$ follows an approximately linear trend,
with an increment of less than 1 between $n=4$ and $n=16$. One
explanation for this observation is that for $m=2$, a field size of
$q=2^8$ is already sufficient for making ARCNC approximately the
same as RLNC.

\begin{figure}[t!]
 \centering
 \includegraphics[width=8.5cm]{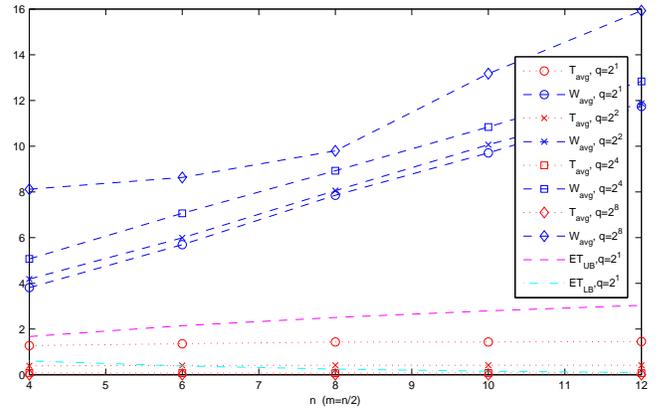}\\
 \caption{Average first decoding time and average memory use. $n=2m$, n increases, field size $q$ also increases. Also plotted are the computed upper and lower bounds on $T_{\text{avg}}$ for $q=2$.}
 \label{fig:combination_variedM_nOver2_qq}
\end{figure}

\subsubsection{$n=2m$}
Fig.~\ref{fig:combination_variedM_nOver2_qq} plots
$T_{\text{avg}}$, $W_{\text{avg}}$, and corresponding bounds on
$T_{\text{avg}}$ when $n=2m$, $q=2$. Since $m$
increases with $n$, $ET_{UB}$ and $ET_{LB}$ change with the value of
$n$ as well. Observe that $T_{\text{avg}}$ increases from
approximately 1.27 to approximately 1.45 as $n$ increases from 4 to
12. In other words, even though more sinks are present, with each
sink connected to more intermediate nodes, the majority of sinks are
still able to achieve decodability within very few coding steps.
However, since now $n=2m$, any given sink $r$ is related to all but one other sink; even a single sink requiring additional
coding steps would force almost all sinks to use more memory to
store longer encoding kernels. Compared with Fig.~\ref{fig:combination_variedM_nOver2_qq}, $W_{\text{avg}}$
appears linear in $n$ in this case.

Fig.~\ref{fig:combination_variedM_nOver2_qq} also plots
$T_{\text{avg}}$ and $W_{\text{avg}}$ when $q$ increases. Similar to
the $m=2$ case shown in Fig.~\ref{fig:combination_variedM_nOver2_qq}, $T_{\text{avg}}$ approaches 0 as $q$ becomes larger. $W_{\text{avg}}$ appears linear in $n$ for $q\leq 2^6$, and piecewise linear for $q=2^8$.
This is because $W_{\text{avg}}$ is lower bounded by $\lceil \log_2q
\rceil$. When $n$ becomes sufficiently large, this lower bound is
surpassed, since $q=2^8$ no longer suffices in making all nodes
decode at time 0, thus making ARCNC a good approximation of RLNC.

\subsection{Shuttle Network}\label{subsec:simuShuttle}
\begin{figure}[t!]
  \centering
  \includegraphics[width=8.5cm]{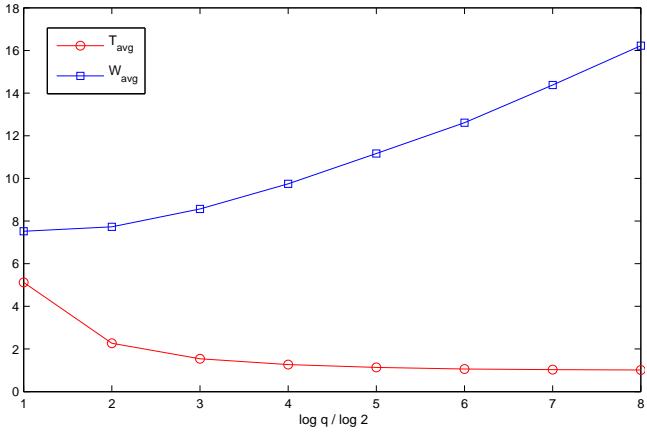}\\
  \caption{Average decoding delay and average code length for the shuttle network as a function of field size.}
  \label{fig:shuttle_qq}
\end{figure}
When there are cycles in the network, as discussed in Section~\ref{sec:algCyclic}, we numerically index edges, and assign local encoding kernels at $t=0$ according to the indices such that no logical contradictions exist in data transmitted around each cycle. Fig.~\ref{fig:shuttle_qq} plots $T_{\text{avg}}$ and $W_{\text{avg}}$ for the shuttle network, with the index assignment given in Fig.~\ref{fig:shuttlenet}. As discussed in the example shown in Figs.~\ref{fig:shuttle0} and \ref{fig:shuttle1}, with this edge index index, both $r_1$ and $r_2$ require at least 2 time steps to achieve decodability. This conclusion is verified by the plot shown in Figure~\ref{fig:shuttle_qq}. As field size $q$ increases, $T_{\text{avg}}$ converges to 1, while $W_{\text{avg}}$ converges to $2\log_2q$. When $q=2^1$, $T_{\text{avg}}$ is 5.1.

\subsection{Acyclic and Cyclic Random Geometric Networks} \label{subsec:simuAcyclicRandom}

To see the performance of ARCNC in random networks, we use random geometric graphs \cite{newman2003random} as the network model, with added acyclic or cyclic constraints. In random geometric graphs, nodes are put into a geometrically confined area $[0,1]^2$, with coordinates chosen uniformly randomly. Nodes which are within a given distance are connected. Call this distance the connection radius. In our simulations, we set the connection radius to 0.4. The resulting graph is inherently bidirectional. 

For acyclic random networks, we number all nodes, with source as node 1, and sinks as nodes with the largest numbers. A node is allowed to transmit to only nodes with numbers larger than its own. An intermediate node on a path from the source to a sink can be a sink itself. To ensure the max-flow to each receiver is non-zero, one can choose the connection radius to make the graph connected with high probability; we fix this value to $0.4$, and throw away instances where at least one receiver is not connected to the source. Once an acyclic random geometric network is generated, we use the smallest min-cut over all sinks as the source symbol rate, which is the number of source symbols generated at each time instant.

Figs.~\ref{fig:random_acyclic_and_cyclic_total25_rx2to12_qq2_sinksLast}
and \ref{fig:random_acyclic_and_cyclic_total10to45_rx3_qq2_sinksLast}
plot the average first decoding time $T_\text{avg}$ and average memory use $W_\text{avg}$ in acyclic random geometric networks.
Fig.~\ref{fig:random_acyclic_and_cyclic_total25_rx2to12_qq2_sinksLast}
shows the case where there are 25 nodes within the network, with more counted as sinks, while Fig.~\ref{fig:random_acyclic_and_cyclic_total10to45_rx3_qq2_sinksLast}
shows the case where the number of sinks is fixed to 3, but more nodes are added to the network. In both cases, $T_{\text{avg}}$ is less than 1, indicating that decodability is achieved in 2 steps with high probability. In Fig.~\ref{fig:random_acyclic_and_cyclic_total25_rx2to12_qq2_sinksLast}, the dependence of $W_{\text{avg}}$ on the number of sinks is not very strong, since there are few sinks, and each node is connected to only a small portion of all nodes. In Fig.~\ref{fig:random_acyclic_and_cyclic_total10to45_rx3_qq2_sinksLast}, $W_{\text{avg}}$ grows as the number of nodes increases, since on average, each node is connected to more neighboring nodes, thus its memory use is more likely to be affected by other sinks.

\begin{figure}[t!]
  \centering
  \includegraphics[width=8.6cm]{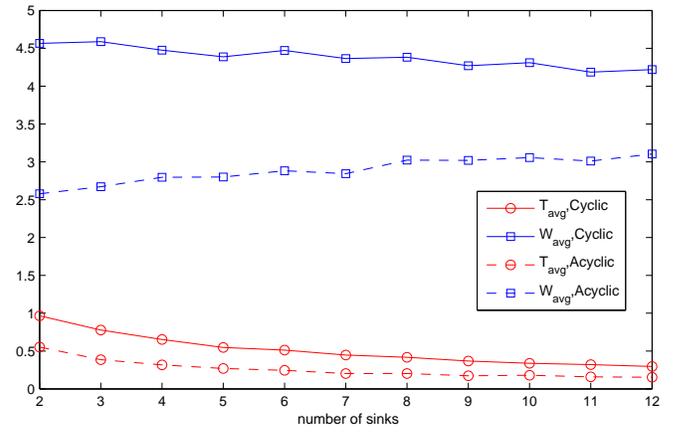}\\
  \caption{Average first decoding time and average memory use in acyclic and cyclic random geometric graphs with 25 nodes, as a function of the number of receivers. Field size is $q=2^2$.}
  \label{fig:random_acyclic_and_cyclic_total25_rx2to12_qq2_sinksLast}
\end{figure}

\begin{figure}[t!]
  \centering
  \includegraphics[width=8.5cm]{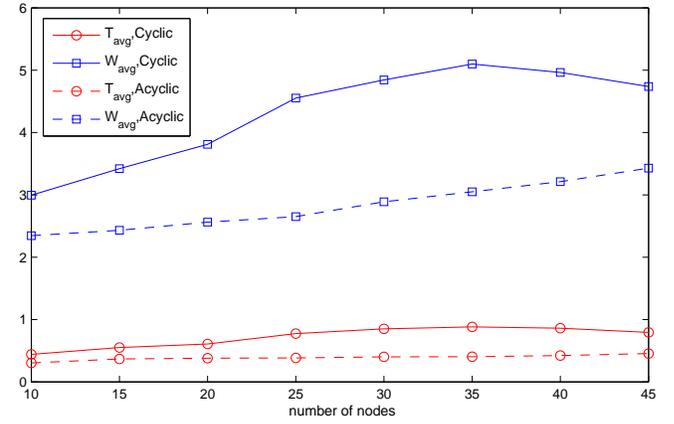}\\
  \caption{Average first decoding time and average memory use in a cyclic random geometric graph with 3 receivers, as a function of the total number of nodes in the network. Field size is $q=2^2$.}
  \label{fig:random_acyclic_and_cyclic_total10to45_rx3_qq2_sinksLast}
\end{figure}

To see the performance of ARCNC in cyclic random networks, note that
random geometric graphs are inherently bidirectional. We apply the
following modifications to make the network cyclic. First, we number
all nodes, with source as node 1, and sinks as the nodes with the
larges numbers. Second, we replace each bidirectional edge with 2
directed edges. Next, a directed edge from a lower numbered to a
higher numbered node is removed from the graph with probability 0.2,
and a directed edge from a higher numbered to a lower numbered node
is removed from the graph with probability 0.8. Such edge removals
ensure that not all neighboring node pairs form cycles, and cycles
can exist with positive probabilities. We do not consider other edge
removal probabilities in our simulations. The effect of random graph
structure on the performance of ARCNC is a non-trivial problem and
will not be analyzed in this paper.

Figs.~\ref{fig:random_acyclic_and_cyclic_total25_rx2to12_qq2_sinksLast}
and
\ref{fig:random_acyclic_and_cyclic_total10to45_rx3_qq2_sinksLast}
also plot the average first decoding time and average memory use in
cyclic random geometric networks. Again, in both cases, the average
first decoding time $T_{\text{avg}}$ is less than 1, indicating that
decodability is achieved in 2 steps with high probability.
$W_{\text{avg}}$ stays approximately constant when more nodes
becomes sinks. On the other hand, when the number of sinks is fixed
to 3, while more nodes are added to the network, $W_{\text{avg}}$
first increases, then decreases in value. This is because as more
nodes are added, since the connection radius stays constant at 0.4,
each node is connected to more neighboring nodes. Sharing parents
with more nodes first increase the memory use of a given node.
However, as more nodes are added and more cycles form, edges are
utilized more efficiently, thus bringing down both $T_{\text{avg}}$
and $W_{\text{avg}}$. Note that when compared with the acyclic case,
cyclic networks with the same number of nodes or same number of
sinks require longer decoding time as well as more memory. This is
expected, since with cycles, sinks are related to more nodes in
general.

\section{Conclusion}\label{sec:conclusion}
We propose an adaptive random convolutional network code (ARCNC),
operating in a small field, locally and automatically adapting to the network topology by incrementally growing the constraint length of the convolution process. Through analysis and simulations, we show that ARCNC performs no worse than random algebraic linear network codes in terms of decodability, while bringing significant gains in terms of decoding delay and memory use in some networks. There are three main advantage of ARCNC compared with scalar network codes and conventional convolutional network codes. Firstly, it operates in a small finite field, reducing the computation overheads of encoding and decoding operations. Secondly, it adapts to unknown network topologies, both acyclic and cyclic. Lastly, it allows codes of different constraint lengths to co-exist within a network, thus bringing practical gains in terms of smaller decoding delays and reduced memory use. The amount of gains achievable through ARCNC is dependent on the number of sinks connected to each other through mutual ancestors. In practical large-scale networks, the number of edges connected to an intermediate node or a sink is always bounded. Thus ARCNC could be beneficial in most general cases. 

One possible extension of this adaptive algorithm is to consider its use with other types of connections such as multiple multicast, or multiple unicast. Another possible direction of future research is to understand the impact of memory used during coding on the rates of innovative data flow through paths along the network. ARCNC presents a feasible solution to the multicast problem, offering gains in terms of delay and memory, but it is not obvious whether a constraint can be added to memory, while jointly optimizing rates achievable at sinks through this adaptive scheme.


\bibliographystyle{IEEEtranTCOM}
\bibliography{references}

\end{document}